\documentclass{IEEEtran}
\pdfoutput=1

\usepackage{amsmath}
\usepackage{amsthm}
\usepackage{amssymb}
\usepackage{bm}
\usepackage{xspace}
\usepackage{xcolor}
\usepackage{graphicx}
\usepackage{url}
\usepackage{framed}
\usepackage{float}
\usepackage{rotating}
\usepackage{verbatim}
\usepackage{listings}
\usepackage{lscape}
\usepackage[normalem]{ulem}

\usepackage{pgfplots}
\usepackage{pgf}
\usepackage{tikz}
\usepackage{cite}
\usetikzlibrary{arrows,shapes.misc,chains,scopes}
\usetikzlibrary{matrix}
\usetikzlibrary{calc}
\usetikzlibrary{fit}
\usetikzlibrary{graphs}

\usepackage{pgfplotstable}
\usepackage{booktabs}
\usepackage{colortbl}

\newcommand{\executeiffilenewer}[3]{%
\ifnum\pdfstrcmp{\pdffilemoddate{#1}}%
{\pdffilemoddate{#2}}>0%
{\immediate\write18{#3}}\fi%
}

\graphicspath{{figures/}}

\theoremstyle{plain}
\newtheorem{proposition}{Proposition}

\newtheorem{definition}{Definition}

\newcounter{algocount}
\newcounter{examplecount}

\newenvironment{example}{\refstepcounter{examplecount}\begin{trivlist}\item \textbf{Example \theexamplecount.}}{\end{trivlist}}

\newcommand{\matp}{\boldsymbol{P}}

\newcommand{\capacity}{\ensuremath{\mathsf{C}}\xspace}

\newcommand{\bmm}{\begin{matrix}}
\newcommand{\emm}{\end{matrix}}
\newcommand{\bpm}{\begin{pmatrix}}
\newcommand{\epm}{\end{pmatrix}}

\newcommand{\bsbm}{\left[\begin{smallmatrix}}
\newcommand{\esbm}{\end{smallmatrix}\right]}
\newcommand{\bspm}{\left(\begin{smallmatrix}}
\newcommand{\espm}{\end{smallmatrix}\right)}

\newcommand{\bbm}{\begin{bmatrix}}
\newcommand{\ebm}{\end{bmatrix}}

\DeclareMathOperator*{\argmin}{argmin}

\DeclareMathOperator{\entop}{\mathbb{H}}
\DeclareMathOperator{\miop}{\mathbb{I}}

\DeclareMathOperator*{\st}{subject\;to}

\newcommand{\mset}[1]{\mathcal{#1}}
\newcommand{\rv}[1]{\mathsf{#1}}
\newcommand{\pmf}[1]{P_{\rv{#1}}}
\newcommand{\pmfn}[2]{P_{\rv{#1}^{#2}}}
\newcommand{\pmfd}[2]{P_{\rv{#1}_{#2}}}

\newcommand{\entrp}[1]{\mathbb{H}\left( #1 \right)}

\newcommand{\ber}[1]{Bernoulli$\left(#1\right)$ }
\newcommand{\beronehalf}{\ber{\frac{1}{2}}}

\newcommand{\diverg}[2]{\mathbb{D}\left( #1 \vert\vert #2 \right)}

\newcommand{\nTypePA}[2]{\mset{T}_{#2}^{#1}}
\newcommand{\nTypePAcard}{\vert \nTypePA{n}{\pmf{\bar{A}}} \vert}

\title{Constant Composition Distribution Matching}

\author{Patrick Schulte and Georg B\"ocherer,~\IEEEmembership{Member,~IEEE}
\thanks{This work was supported by the German Ministry of Education and Research in the framework of an Alexander von Humboldt Professorship.}
\thanks{The authors are with the Institute for Communications Engineering, Technische Universit\"at M\"unchen, M\"unchen, Germany.}
}

\DeclareMathOperator{\supp}{supp}

\usetikzlibrary{external} 

\newcommand{\power}{\mathsf{P}}

\pgfplotsset{
compat=newest,
width=\columnwidth,
height=0.3\textheight
}
\begin{document}
\maketitle
\begin{abstract}
Distribution matching transforms independent and \beronehalf distributed input bits into a sequence of output symbols with a desired distribution. 
Fixed-to-fixed length,  invertible, and low complexity encoders and decoders based on constant composition and arithmetic coding are presented.
Asymptotically in the blocklength, the encoder achieves the maximum rate, namely the entropy of the desired distribution.
Furthermore, the normalized divergence of the encoder output and the desired distribution goes to zero in the blocklength.
\end{abstract}
\section{Introduction}
\IEEEPARstart{A}{}\emph{distribution matcher} transforms independent Bernoulli$(\frac{1}{2})$ distributed input bits into output symbols with a desired distribution.
We measure the distance between the matcher output distribution and the desired distribution by normalized \emph{informational divergence}\cite[p.~7]{csiszar2011information}.
Informational divergence is also known as \emph{Kullback-Leibler divergence} or \emph{relative entropy} \cite[Sec.~2.3]{cover2006elements}.
A dematcher performs the inverse operation and recovers the input bits from the output symbols.
A distribution matcher is a building block of the bootstrap scheme \cite{bocherer2011operating} that achieves the capacity of arbitrary discrete memoryless channels \cite{mondelli2014achieve}.
Distribution matchers are used in  \cite[Sec.~VI]{MacKay1999} for rate adaption and in \cite{bocherer2015bandwidth} to achieve the capacity of the additive white Gaussian noise channel.

Prefix-free distribution matching was proposed in \cite[Sec.~IV.A]{forney1984efficient}. In \cite{kschischang1993optimal,Ungerboeck2002} Huffman codes are used for matching. Optimal variable-to-fixed and fixed-to-variable length distribution matchers are proposed in \cite{bocherer2011matching} and \cite{amjad2013fixed}, respectively.
The codebooks of the matchers in \cite{kschischang1993optimal,Ungerboeck2002,bocherer2011matching,amjad2013fixed} must be generated offline and stored. This is infeasible for large codeword lengths, which are necessary to achieve the maximum rate.
This problem is solved in \cite{Cai2007,bocherer2013arithmetic} by using arithmetic coding to calculate the codebook online. The matchers proposed in \cite{Cai2007,bocherer2013arithmetic} are asymptotically optimal. All approaches \cite{kschischang1993optimal,Ungerboeck2002,bocherer2011matching,amjad2013fixed,Cai2007, bocherer2013arithmetic} are variable length, which can lead to varying transmission rate, large buffer sizes, error propagation and synchronization problems \cite[Sec.~I]{kschischang1993optimal}.
\emph{Fixed-to-fixed} (f2f) length codes do not have these issues.
The author of \cite[Sec.~4.8]{amjad2013algorithms} suggests to concatenate short codes and the authors of \cite{mondelli2014achieve} employ a forward error correction decoder to build an f2f length matcher.
The dematchers of \cite{mondelli2014achieve,amjad2013algorithms} cannot always recover the input sequence with zero error.
Hence systematic errors are introduced that cannot be corrected by the error correction code or by retransmission.
The thesis  \cite{schulte2014Zero} proposes an invertible f2f length distribution matcher called \emph{adaptive arithmetic distribution matcher} (aadm). The algorithm is computationally complex.

In this work we propose practical, invertible, f2f length distribution matchers.
They are asymptotically optimal and are based on constant composition codes indexed by arithmetic coding.
The paper is organized as follows.
In Section \ref{sec:problem} we formally define distribution matching.
We analyze constant composition codes in Section \ref{sec:ConstanCompositionDistributionMatching}.
In Section \ref{sec:ArithmeticCoding} we show how a constant composition distribution matcher (ccdm) and dematcher can be implemented efficiently by arithmetic coding.
\begin{figure}
\centering
\usetikzlibrary{fit}
\tikzset{%
 dmsblock/.style = {draw,thick,top color = gray!1,bottom color = gray!10,rectangle,text centered,minimum height = 25,minimum width = 25},
 dsblock/.style = {draw,thick,top color = gray!1,bottom color = gray!10,rectangle,text centered,text width = 30,minimum height = 30, dash pattern=on 1pt off 4pt on 6pt off 4pt}
}
\tikzset{dashdot/.style={draw=black!50!white, line width=1pt,
                               dash pattern=on 1pt off 4pt on 6pt off 4pt,
                                inner sep=1.2mm, rectangle, rounded corners}}

\begin{tikzpicture}[scale= 0.9,auto, thick, node distance=39] 
\node (xes) {\footnotesize$\rv{B}^{m}$};
	\node [dmsblock, right of = xes, node distance=50](matcher){\footnotesize Matcher};
	\node [right of = matcher, minimum height = 40,node distance=50] (out_dev){\footnotesize$\tilde{\rv{A}}^{n}$};
    \node [dmsblock, right of = out_dev, node distance=55] (dematcher){\footnotesize Dematcher};
	\node [right of = dematcher, minimum height = 40,node distance=50] (recover){\footnotesize$\rv{B}^{m}$};
	\node [dmsblock,below of= matcher](pa){\footnotesize$P_{\rv{A}}$};
	\node [below of=out_dev, minimum height = 40](AAs){\footnotesize$\rv{A}^{n}$};
	\node [dashdot,fit={(matcher) (xes)}](ptilde){};
	
	\node[above] at (ptilde.north) {\footnotesize$P_{\mathsf{\tilde{\rv{A}}}^n}$};
	\draw[-latex] (xes) -- (matcher);
	\draw[] (matcher) -- (out_dev);
	\draw[-latex] (out_dev) -- (dematcher);
	\draw[] (dematcher) -- (recover);
	\draw[](pa) -- (AAs);
\end{tikzpicture}
\caption{Matching a data block $\rv{B}^m = \rv{B}_1\!\ldots\!\rv{B}_m$ to output symbols $\tilde{\rv{A}}^n = \tilde{\rv{A}}_1\!\ldots\!\tilde{\rv{A}}_n$ and reconstructing the original sequence at the dematcher.
The rate is $\frac{m}{n}\left[ \frac{\text{bits}}{\text{output symbol}} \right]$. The matcher can be interpreted as emulating a discrete memoryless source $\pmf{A}$.}
\label{fig:memoryless_Source}
\end{figure}
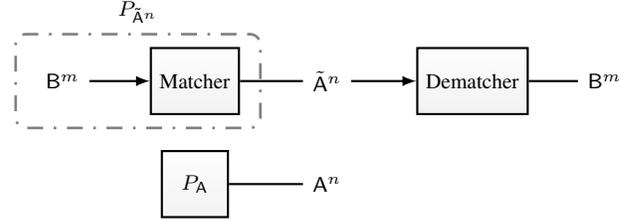
\section{Problem statement} \label{sec:problem}
The entropy of a discrete random variable $\rv{A}$ with alphabet $\mset{A}$ and distribution $\pmf{A}$ is
\begin{equation}
\entrp{\rv{A}} = \sum_{a \in \supp(\pmf{A})} -\pmf{A}(a) \log_2 {\pmf{A}(a)}
\end{equation}
where $\supp(\pmf{A}) \subseteq \mset{A}$ is the support of $\pmf{A}$.
The informational divergence of two distributions on $\mset{A}$ is
\begin{equation}
 \diverg{\pmf{\hat{A}}}{\pmf{A}} = \sum_{a \in \supp(\pmf{\hat{A}})} \pmf{\hat{A}}(a) \log_2 \frac{\pmf{\hat{A}}(a)}{\pmf{A}(a)}.
\end{equation}
The normalized informational divergence for length $n$ random vectors $\rv{\hat{A}}^n = \rv{\hat{A}}_1\!\ldots\!\rv{\hat{A}}_n$ and $\rv{A}^n$ is defined as
\begin{equation}
 \frac{\diverg{\pmfn{\hat{A}}{n}}{\pmfn{A}{n}}}{n}.
\end{equation}
For  random vectors with independent and identically distributed (iid) entries, we write
\begin{equation}
 \pmf{A}^n(a^n) = \prod_{i=1}^n \pmf{A}(a_i).
\end{equation}

A one-to-one f2f distribution matcher is an invertible function $f$. We denote the inverse function by $f^{-1}$. The mapping imitates a desired distribution $\pmf{A}$ by mapping $m$ \beronehalf distributed bits $\rv{B}^m$ to length $n$ strings $\tilde{\rv{A}}^n = f(B^m) \in \mset{A}^n$.
The output distribution is $\pmfn{\tilde{A}}{n}$.
The concept of one-to-one f2f distribution matching is illustrated in Fig. \ref{fig:memoryless_Source}.
\begin{definition} \label{def:achievableRate}
A matching rate $R = m/n$ is achievable for a distribution $\pmf{A}$ if for any $\alpha > 0$ and sufficiently large $n$ there is an invertible mapping $f\colon\lbrace0,1\rbrace^m \to \mset{A}^n$ for which
\begin{equation}
\frac{\diverg{P_{f(\rv{B}^m)}}{\pmf{A}^n}}{n}\leq \alpha. \label{eq:ratecond:diverg}
\end{equation} 
\end{definition}
The following proposition in \cite{bocherer2014informational} relates the rate $R$ and \eqref{eq:ratecond:diverg}.
\begin{proposition}[Converse,{\cite[Proposition~8]{bocherer2014informational}}] \label{prop:converse}
	 There exists a positive-valued function $\beta$ with
\begin{equation}
 \beta(\alpha) \overset{\alpha \rightarrow 0}{\longrightarrow} 0
\end{equation}
such that \eqref{eq:ratecond:diverg} implies
\begin{equation}
 \frac{m}{n} \leq \frac{\entrp{\rv{A}}}{\entrp{\rv{B}}} + \beta(\alpha).
\end{equation}
\end{proposition}
Proposition \ref{prop:converse} bounds the maximum rate that can be achieved under condition \eqref{eq:ratecond:diverg}.
Since $\entrp{\rv{B}} = 1$ we have
\begin{equation}
 R \leq \entrp{\rv{A}} \label{eq:conv:result}
\end{equation}
for any achievable rate $R$.
\section{Constant Composition Distribution Matching}
\label{sec:ConstanCompositionDistributionMatching}
The empirical distribution of a vector $\boldsymbol{c}$ of length $n$ is defined as
\begin{equation}
 P_{\mathsf{\bar{A}},\boldsymbol{c}}(a) := \frac{n_a(\boldsymbol{c})}{n}
\end{equation}
where $n_a(\boldsymbol{c}) = \left| \left\lbrace i: c_i = a \right\rbrace \right|$ is the number of times symbol $a$ appears in $\boldsymbol{c}$.
The authors of \cite[Sec.~2.1]{csiszar2004information} call $P_{\mathsf{\bar{A}},\boldsymbol{c}}$ the \emph{type} of $\boldsymbol{c}$. An $n$-type is a type based on a length $n$ sequence.
A codebook $\mset{C}_{\!\text{ccdm}} \subseteq{\mset{A}}^{n}$ is called a \emph{constant composition code} if all codewords are of the same type, i.e., 
 $n_a(\boldsymbol{c})$ does not depend on the codeword $\boldsymbol{c}$.
 We will write $n_a$ in place of $n_a(\boldsymbol{c})$ for a constant composition code.

\subsection{Approach}
We use a constant composition code with $n_a \approx \pmf{A} n$.
As all $n_a$ need to be integers and add up to $n$, there are multiple possibilities to choose the $n_a$.
We use the allocation that solves
\begin{equation}
\begin{split}
 \pmf{\bar{A}} = &\argmin_{\pmf{\bar{A}'}} \: \diverg{\pmf{\bar{A}'}}{\pmf{A}}\\
 &\st\pmf{\bar{A}'} \text{ is }n\text{-type}.
\end{split}\label{eq:quantization}
\end{equation}
The solution of \eqref{eq:quantization} can be found efficiently by \cite[Algorithm~2]{bocherer2014optimal}.
Suppose the output length $n$ is fixed and that we can choose the input length $m$.
Let $\nTypePA{n}{\pmf{\bar{A}}}$ be the set of vectors of type $\pmf{\bar{A}}$, i.e., we have
\begin{equation}
 \nTypePA{n}{\pmf{\bar{A}}} =\lbrace\boldsymbol{v} \; \vert \; \boldsymbol{v} \in \mset{A}^n, \frac{n_a(\boldsymbol{v})}{n} = \pmf{\bar{A}}(a)\quad \forall a \in \mset{A}\rbrace.
\end{equation}
The matcher is invertible, so we need at least as many codewords as input blocks.
The input blocklength must thus not exceed
$ \log_2 \nTypePAcard $. We set the input length to $m = \lfloor\log_2 \nTypePAcard \rfloor$
and we define the encoding function
\begin{equation}
f_{\!\text{ccdm}}: \lbrace 0,1 \rbrace^m \to \nTypePA{n}{\pmf{\bar{A}}}.
\end{equation}
The actual mapping $f_{\!\text{ccdm}}$ can be implemented efficiently by arithmetic coding, as we will show in Section~\ref{sec:ArithmeticCoding}. The constant composition codebook is now given by the image of $f_{\!\text{ccdm}}$, i.e.,
\begin{equation}
\mset{C}_{\!\text{ccdm}} = f_{\!\text{ccdm}}(\lbrace 0,1 \rbrace^m).
\end{equation}
Since  $f_{\!\text{ccdm}}$ is invertible, the codebook size is $|\mset{C}_{\!\text{ccdm}}|=2^m$.

\subsection{Analysis} \label{sec:analysis}
We show that $f_{\!\text{ccdm}}$ asymptotically achieves all rates satisfying \eqref{eq:conv:result}.
We can bound $m$ by
\begin{equation}
m = \left\lfloor\log_2 \nTypePAcard \right\rfloor \geq \log_2 \nTypePAcard -1.
\label{eq:bound_kc}
\end{equation}
Recall that the matcher output distribution is $\pmfn{\tilde{A}}{n}$.
We have
\begin{align}
&\diverg{\pmfn{\tilde{A}}{n}}{\pmf{A}^{n}} = \sum_{a^{n} \in \mathcal{C}_{\!\text{ccdm}} \subseteq \nTypePA{n}{\bar{P}_{\rv{A}}}} 2^{-m} \log_2\frac{2^{-m}}{\pmf{A}^n(a^n)}\frac{\pmf{\bar{A}}(a^n)}{\pmf{\bar{A}}(a^n)}\nonumber \displaybreak[3]\\
&\quad= \diverg{\pmfn{\tilde{A}}{n}}{\pmf{\bar{A}}^{n}} + \sum_{a^{n} \in \mathcal{C}_{\!\text{ccdm}} \subseteq \nTypePA{n}{\bar{P}_{\rv{A}}}} 2^{-m} \log_2\frac{\pmf{\bar{A}}^n(a^n)}{\pmf{A}^n(a^n)} \nonumber \displaybreak[3]\\
&\quad= \diverg{\pmfn{\tilde{A}}{n}}{\pmf{\bar{A}}^{n}} + \vert \mathcal{C}_{\!\text{ccdm}} \vert 2^{-m} \sum_{a \in \mset{A}} n_a \log_2 \frac{\pmf{\bar{A}}(a)}{\pmf{A}(a)}\nonumber \\
&\quad= \underbrace{\diverg{\pmfn{\tilde{A}}{n}}{\pmf{\bar{A}}^{n}}}_{\text{Term 1}}+ n \underbrace{\diverg{\pmf{\bar{A}}}{\pmf{A}}}_{\text{Term 2}}.\label{eq:divergenceCCC}
\end{align}
For Term 1 we obtain
\begin{align}
\diverg{\pmfn{\tilde{A}}{n}}{\pmf{\bar{A}}^{n}} &= \sum_{a^{n} \in \mathcal{C}_{\!\text{ccdm}} \subseteq \nTypePA{n}{\bar{P}_{\rv{A}}}} 2^{-m}\log_2 \frac{2^{-m}}{\prod\limits_{i \in \mset{A}}\pmf{\bar{A}}(i)^{n_i}}\notag\\
&= \sum\limits_{\mathcal{C}_{\!\text{ccdm}}} 2^{-m} \log_2 \frac{2^{-m}}{2^{-n\mathbb{H}(\rv{\bar{A}})}}\notag\\
&= n\mathbb{H}(\rv{\bar{A}})-m.\label{eq:asymptoticDivergence}
\end{align}
Using \eqref{eq:asymptoticDivergence} in \eqref{eq:divergenceCCC} and dividing by $n$ we have
\begin{equation}
 \frac{\diverg{\pmfn{\tilde{A}}{n}}{\pmf{A}^{n}}}{n} = \mathbb{H}(\rv{\bar{A}}) - R +
\diverg{\pmf{\bar{A}}}{\pmf{A}}. \label{eq:normdivergExplicit}
\end{equation}
The choice \eqref{eq:quantization} of $\pmf{\bar{A}}$ guarantees (see \cite[~Proposition 4]{bocherer2014optimal}) that for the third term in \eqref{eq:normdivergExplicit} we have
\begin{equation}
 \diverg{\pmf{\bar{A}}}{\pmf{A}} < \frac{k}{\displaystyle \min_{a \in \supp\pmf{A}}\pmf{A}(a) n^2}
 \label{eq:upperboundOnChoiceOfPABar}
\end{equation}
where $k = |\mset{A}|$ is the alphabet size. 
Consequently, we know that this term vanishes as the blocklength approaches infinity, i.e., we have
\begin{equation}
 \lim_{n\rightarrow \infty} \diverg{\pmf{\bar{A}}}{\pmf{A}} = 0. \label{eq:probbilityidentity}
\end{equation}
We now relate the input and output lengths to understand the asymptotic behavior of the rate. 
By \cite[~Lemma 2.2]{csiszar2004information}, we have
\begin{equation}
\nTypePAcard \geq \binom{n + k - 1}{k - 1}^{-1} 2^{n\mathbb{H}(\rv{\bar{A}})} \geq (n+k)^{-k} 2^{n\mathbb{H}(\rv{\bar{A}})}.
\end{equation}
Taking the logarithm to the base $2$ and dividing by $n$ we have
\begin{equation}
\frac{\log_2{\nTypePAcard}}{n} \geq \frac{-k \log_2(n+k)}{n}+ \entrp{\rv{\bar{A}}}. \label{eq:bound_log}
\end{equation}
For the rate, we obtain
\begin{align}
R = \frac{m}{n} &\overset{\eqref{eq:bound_kc}}{\geq} \frac{\log_2 \nTypePAcard}{n} - \frac{1}{n} \notag\\
&\overset{\eqref{eq:bound_log}}{\geq} \frac{-k \log_2(n+k)}{n}+ \entrp{\rv{\bar{A}}} - \frac{1}{n} \label{eq:lowerBoundOnRate}
\end{align}
and in the asymptotic case
\begin{equation}
\lim_{n \rightarrow \infty} R = \entrp{\rv{\bar{A}}}. \label{eq:asymptoticRate}
\end{equation}
 From \eqref{eq:probbilityidentity} and \cite[Proposition~6]{bocherer2014informational} we know that $\entrp{\rv{\bar{A}}} \rightarrow \entrp{\rv{A}}$ and by  \eqref{eq:probbilityidentity} and \eqref{eq:asymptoticRate} in \eqref{eq:normdivergExplicit}, normalized divergence approaches zero for $n \rightarrow \infty$.
\begin{figure}
\centering
%
%
\begin{tikzpicture}[font=\footnotesize]

\begin{axis}[%
width=0.36\textwidth,
height=65mm,
scale only axis,
xmode=log,
xmin=10,
xmax=10000,
xminorticks=true,
xmajorgrids,
xminorgrids,
xlabel={Output blocklength $n$ (in symbols)},
separate axis lines,
every outer y axis line/.append style={blue},
every y tick label/.append style={font=\color{blue}\footnotesize},
ymax=1,
ymode=log,
ylabel={Normalized divergence (in bits per symbol)},
ymajorgrids,
legend style={font=\footnotesize,at={(0.98,0.98)},anchor=north east, legend cell align=left}
]
\addlegendimage{empty legend}
\addlegendentry{\hspace{-0.5cm}\textbf{Normalized divergence}}
\addplot [color=blue,thick, mark=x, mark repeat={4}, mark size = 2]
  table[row sep=crcr]{%
10	0.562126661727604\\
12	0.453382131629775\\
13	0.432315402813148\\
15	0.362627946590146\\
18	0.32526150119509\\
20	0.333701299189644\\
23	0.314842037123781\\
27	0.254814991792944\\
31	0.241122814776856\\
36	0.218374085844941\\
41	0.197465216472161\\
47	0.181067985474743\\
54	0.161410631702398\\
63	0.138753521727418\\
72	0.132928904631127\\
83	0.118457253686564\\
95	0.103286820918822\\
110	0.0931969744194568\\
126	0.0816432955981793\\
146	0.0724207118137824\\
168	0.0656126989642668\\
193	0.0574280789747415\\
222	0.0528788627834519\\
256	0.0461917621521581\\
295	0.0405077728021532\\
339	0.0355300898858071\\
391	0.0330986077261811\\
450	0.0290849444668692\\
518	0.0267730022617277\\
596	0.0225935577896297\\
687	0.0199173303496859\\
791	0.0181652321707025\\
910	0.0160941358132567\\
1048	0.0146691836020913\\
1207	0.0127348302270843\\
1389	0.0112698186118978\\
1600	0.00984009835184821\\
1842	0.00877160200036727\\
2121	0.00757741433384309\\
2442	0.00665916840780431\\
2812	0.00602178166025024\\
3237	0.00538185939785032\\
3728	0.00466992737768425\\
4292	0.00414829300749415\\
4942	0.00364836980728526\\
5690	0.00324842695904912\\
6551	0.00283984280494188\\
7543	0.00257743889115127\\
8685	0.00222171632168395\\
10000	0.00201427300066644\\
};
\addlegendentry{ccdm};
\addplot [color=blue!40,ultra thick,solid]
table[row sep=crcr]{%
14	0.107798993530132\\
21	0.0883559070013077\\
28	0.0755378711843144\\
35	0.066692316179804\\
42	0.0599780246406431\\
49	0.0546781111145991\\
56	0.0503595610610835\\
63	0.0467659921881641\\
70	0.0437277614995217\\
77	0.0411110358454136\\
84	0.0388297118821924\\
91	0.0368201330727669\\
};
\addlegendentry[align = left]{ optimal f2f, $R = \mathbb{H}(A)$  \cite[Sec.~4.4]{amjad2013algorithms}};

\addplot [color=red!50!yellow,solid,thick]
  table[row sep=crcr]{%
10	nan\\
32	0.495066622394646\\
100	0.103716690701949\\
316	0.0306774258134604\\
1000	0.0104680836591021\\
3162	0.00316181744686671\\
5623	0.00206549076999741\\
10000	0.00108091104851264\\
};
\addlegendentry[align = left]{aadm, $R = \mathbb{H}(A)$ \cite{schulte2014Zero}};

\end{axis}

\begin{axis}[%
width=0.36\textwidth,
height=60mm,
scale only axis,
xmin=10,
xmax=10000,
xmode=log,
every outer y axis line/.append style={black!50!green},
every y tick label/.append style={font=\color{black!50!green}\footnotesize},
ymin=1,
ymax=2,
ytick={  1, 1.5,   2},
y label style={at={(axis description cs:1.08,.5)},rotate=180,anchor=south},
ylabel={Rate (in bits per symbol)},
legend style={font=\footnotesize,at={(0.02,0.02)},anchor=south west, legend cell align=left},
axis x line*=bottom,
axis y line*=right
]
\addlegendimage{empty legend}
\addlegendentry{\hspace{-0.5cm}\textbf{Rate}}

\addplot [color=black!50!green,thick, mark=o, mark repeat={4}, mark size = 2]
  table[row sep=crcr]{%
10	1.3\\
12	1.33333333333333\\
13	1.30769230769231\\
15	1.33333333333333\\
18	1.38888888888889\\
20	1.35\\
23	1.47826086956522\\
27	1.48148148148148\\
31	1.48387096774194\\
36	1.55555555555556\\
41	1.5609756097561\\
47	1.5531914893617\\
54	1.59259259259259\\
63	1.61904761904762\\
72	1.61111111111111\\
83	1.63855421686747\\
95	1.65263157894737\\
110	1.65454545454545\\
126	1.66666666666667\\
146	1.68493150684932\\
168	1.68452380952381\\
193	1.69430051813472\\
222	1.6981981981982\\
256	1.70703125\\
295	1.70847457627119\\
339	1.71091445427729\\
391	1.71611253196931\\
450	1.72222222222222\\
518	1.72200772200772\\
596	1.72818791946309\\
687	1.73216885007278\\
791	1.73198482932996\\
910	1.73406593406593\\
1048	1.73568702290076\\
1207	1.73736536868268\\
1389	1.73866090712743\\
1600	1.74125\\
1842	1.74158523344191\\
2121	1.74257425742574\\
2442	1.74324324324324\\
2812	1.74395448079659\\
3237	1.74482545566883\\
3728	1.74543991416309\\
4292	1.74603914259087\\
4942	1.74645892351275\\
5690	1.7469244288225\\
6551	1.74736681422684\\
7543	1.7477131114941\\
8685	1.74795624640184\\
10000	1.7481\\
};
\addlegendentry{ccdm};

\addplot [color=green!30!black,thick]
  table[row sep=crcr]{10	1.750114273000667\\
  10000	1.750114273000667\\
  };
 \addlegendentry{$\mathbb{H}(A)$};

\addplot [color=green!70!yellow,dashed, ultra thick,]
  table[row sep=crcr]{10	0.127172304177625\\
12	0.333447606334\\
13	0.415510321846716\\
15	0.550666936082377\\
18	0.70357391330349\\
20	0.783121772856436\\
23	0.879699055233107\\
27	0.979122227017426\\
31	1.05601323855616\\
36	1.13101115134652\\
41	1.18993348311484\\
47	1.24607807368394\\
54	1.29767123632455\\
63	1.34909273709859\\
72	1.38911830003158\\
83	1.42756277922482\\
95	1.46045715215521\\
110	1.49255463612195\\
126	1.51924545353945\\
146	1.5452151307953\\
168	1.56734606455538\\
193	1.58696293995494\\
222	1.60470564304724\\
256	1.6208585259221\\
295	1.63521255538442\\
339	1.64778902502223\\
391	1.65931418253453\\
450	1.66943384200253\\
518	1.67847021125109\\
596	1.68649804017855\\
687	1.69373848318763\\
791	1.7001278707242\\
910	1.70577998572569\\
1048	1.71084358983543\\
1207	1.71534382345736\\
1389	1.71931807561377\\
1600	1.72287062693453\\
1842	1.72600963274519\\
2121	1.72879744673775\\
2442	1.73126708432167\\
2812	1.73345789800972\\
3237	1.73539419477352\\
3728	1.73711455919074\\
4292	1.73863358495035\\
4942	1.73997906739614\\
5690	1.74116859512528\\
6551	1.7422202842569\\
7543	1.74315063118236\\
8685	1.74397266128654\\
10000	1.74469895726379\\
};

 \addlegendentry{Lower bound \eqref{eq:lowerBoundOnRate} on ccdm rate};

\end{axis}
\end{tikzpicture}%
\caption{Normalized divergence and rate of ccdm over output blocklength for $\pmf{A}=(0.0722,0.1654,0.3209,0.4415)$. For comparison, the performance of optimal f2f \cite[Sec.~4.4]{amjad2013algorithms} and aadm \cite{schulte2014Zero} is displayed. Because of limited computational resources, we could calculate the performance of optimal f2f only up to a blocklength of $n = 90$.}
\label{fig:normalizedDivergenceAndRateVsBlocklength}
\end{figure}
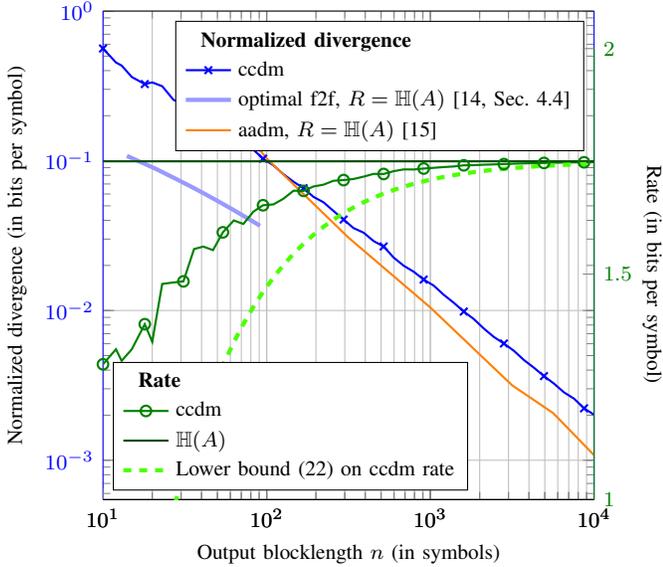
\begin{example} The desired distribution is
	\begin{equation*}
	\pmf{A}=(0.0722,0.1654,0.3209,0.4415).
	\end{equation*}
Fig.~\ref{fig:normalizedDivergenceAndRateVsBlocklength} shows the  normalized divergences and rates of ccdm and the optimal  f2f length matcher \cite[Sec.~4.4]{amjad2013algorithms}. The empirical performance of aadm \cite{schulte2014Zero} is also displayed. For optimal f2f and aadm, the rate is fixed to $\mathbb{H}(\rv{A})$ bits per symbol. Observe that the ccdm needs about 4   times the blocklength of the optimal scheme to reach an informational divergence of $0.06$ bits per symbol.
However, the memory for storing the optimal codebook grows exponentially in $m$.
For $n = 10$, we already need about $10240$ bits = $1.25$ kB; for $n = 100$ we would need $1.441\times10^{19}$ TB of memory.
In this example, ccdm performs better than aadm for short blocklength up to $100$ symbols.
Fig.~\ref{fig:normalizedDivergenceAndRateVsBlocklength}  also shows the lower and upper bounds  \eqref{eq:conv:result} and \eqref{eq:lowerBoundOnRate}, respectively.
\end{example}
\section{Arithmetic Coding}
\label{sec:ArithmeticCoding}
\begin{figure}
	\centering
	 
\usetikzlibrary{decorations.pathreplacing}
\pgfmathsetmacro{\len}{3.}
\begin{tikzpicture}[font=\footnotesize, scale=1.3]
\foreach \x/\y/\col in {1/0011/blue, 2/0101/black, 3/0110/blue,4/1001/blue,5/1010/black,6/1100/blue}{
\pgfmathsetmacro{\a}{\x *100/6};
\pgfmathsetmacro{\b}{\x *\len/6};
\draw[|-|,color=blue,thick] (1,\b-\len/6)-- node[right,color=\col](d\x){\y} (1,\b) node[left,black]{ $\x/6$};
}

\foreach \x/\y in {1/00, 2/01, 3/10, 4/11}{
\pgfmathsetmacro{\a}{\x *100/4};
\pgfmathsetmacro{\b}{\x *\len/4};
\draw[|-|,color=red,thick] (-1,\b-\len/4)-- node[left,color=black](h\x){\y} (-1,\b)node[right,black]{$\x/4$};
\draw[dotted] (-1,\b) --++(2,0);
}

\path[draw,-latex](h1.east) --(d1.west);
\path[draw,-latex](h2.east) --(d3.west);
\path[draw,-latex](h3.east) --(d4.west);
\path[draw,-latex](h4.east) --(d6.west);

\node at (-1.9,3.5/2) {$\lbrace0,1 \rbrace^m$};
\node at (2.1,3.5/2) {$\nTypePA{n}{\bar{P}_{\rv{A}}}$};  
\node at (2.1,3.5/2 -0.6) {\textcolor{blue}{$\mset{C}_{\!\text{ccadm}}$}};   

\end{tikzpicture}
	\caption{Diagram of a constant composition arithmetic encoder with $\pmf{\bar{A}}(0) = \pmf{\bar{A}}(1) = 0.5$, $m=2$ and $n=4$.}
	\label{fig:ccadm_overview}
\end{figure}
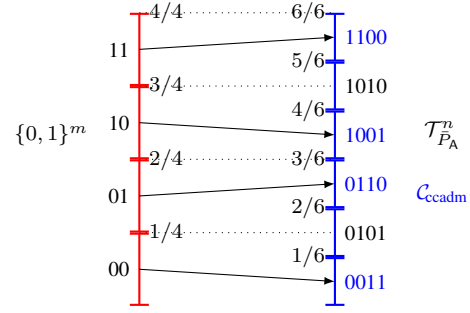

We use arithmetic coding for indexing sequences efficiently.
Our arithmetic encoder associates an interval to each input sequence in $\{0,1\}^m$ and it associates an interval to each output sequence in $\nTypePA{n}{\pmf{\bar{A}}}$, see Fig.~\ref{fig:ccadm_overview} for an example. The size of an interval is equal to the probability of the corresponding sequence according to the input and output model, respectively.
For the input model we choose an iid \beronehalf process. 
We describe the output model by a random vector
\begin{equation}
\rv{\bar{A}}^n = \rv{\bar{A}}_1\rv{\bar{A}}_2\ldots\rv{\bar{A}}_n
\end{equation}
with marginals $\pmfd{\bar{A}}{i}=\pmf{\bar{A}}$ and the uniform distribution
\begin{equation*}
\pmfn{\bar{A}}{n} (a^n)= \frac{1}{\nTypePAcard}\quad\forall a^n\in \nTypePA{n}{\pmf{\bar{A}}}.
\end{equation*}
The intervals are ordered lexicographically.
All input and output intervals range from 0 to 1 because all probabilities add up to 1.
\begin{example}\label{ex:arithmBasic}
Fig.~\ref{fig:ccadm_overview} shows input and output intervals with output length $n=4$ and $\pmf{\bar{A}}(0) = \pmf{\bar{A}}(1) = 0.5$.
There are 4 equally probable input sequences and 6 equally probable output sequences.
The intervals on the input side are $[0,0.25)$, $[0.25, 0.5)$, $[0.5,0.75)$ and $[0.75,1)$. The intervals on the output side are $[0,\frac{1}{6})$, $[\frac{1}{6},\frac{2}{6})$, $[\frac{2}{6},\frac{3}{6})$, $[\frac{3}{6},\frac{4}{6})$, $[\frac{4}{6},\frac{5}{6})$ and $[\frac{5}{6},1)$.\footnote{Please note that in this case no distribution matcher is needed. However, the invertible mapping is of interest in its own right.}
\end{example}

The arithmetic encoder can link an output sequence to an input sequence if the lower border of the output interval is inside the input interval.
In the example (Fig.~\ref{fig:ccadm_overview}) '00' may link to both '0101' and '0011', while for '01' only a link to '0110' is possible.
There are at most two possible choices because by \eqref{eq:bound_kc} the input interval size is less than twice the output interval size.
Both choices are valid and we can perform an inverse operation.
In our implementation, the encoder decides for the output sequence with the lowest interval border.
As a result, the codebook $\mset{C}_{\!\text{ccdm}}$ of Example \ref{ex:arithmBasic} is $\lbrace$'$0011$', '$0110$, '$1001$', '$1100$'$\rbrace$. In general $\mset{C}_{\!\text{ccdm}}$ has cardinality $2^m$ with $2^m \leq \nTypePAcard < 2^{m+1}$ according to \eqref{eq:bound_kc}.
It is not possible to index the whole set $\nTypePA{n}{\pmf{\bar{A}}}$ unless $2^m = \nTypePAcard $.
The analysis of the code (Section \ref{sec:analysis}) is valid for all codebooks $\mset{C}_{\text{\!ccdm}} \subseteq \nTypePA{n}{\pmf{\bar{A}}}$.
The actual subset is implicitly defined by the arithmetic encoder.

We now discuss the online algorithm that processes the input sequentially.
Initially, the input interval spans from 0 to 1.
As the input model is \beronehalf we split the interval into two equally sized intervals and continue with the upper interval in case the first input bit is '1'; otherwise we continue with the lower interval.
After the next input bit arrives we repeat the last step.
After $m$ input bits we reach a size $2^{-m}$ interval.
After every refinement of the input interval the algorithm checks for a sure prefix of the output sequence, e.g., in Fig.~\ref{fig:ccadm_overview} we see that if the input starts with 1 the output must start with 1.
Every time we extend the sure prefix by a new symbol, we must calculate the probability of the next symbol given the sure prefix.
That means we determine the output intervals within the sure interval of the prefix.                                                                                                                                                                                                                                                                                                                                           
The model for calculating the conditioned probabilities is based on drawing without replacement.
There is a bag with $n$ symbols of $k$ discriminable kinds.
$n_a$ denotes how many symbols of kind $a$ are initially in the bag and $n'_a$ is the current number.
The probability to draw a symbol of type $a$ is $n'_a/n$.
If we pick a symbol $a$ both $n$ and $n'_a$ decrement by $1$.

\begin{figure}
\centering
\pgfmathsetmacro{\len}{3}

\begin{tikzpicture}[font=\footnotesize]

\foreach \x/\y/\z/\w in {0/0.16667/$00(11)$/$\frac{1}{6}$,0.16667/0.3333/$010(1)$/$\frac{2}{6}$, 0.3333/0.5/$011(0)$/$\frac{3}{6}$, 0.5/0.6666/$100(1)$/$\frac{4}{6}$ , 0.6666/0.8333/$101(0)$/$\frac{5}{6}$ , 0.8333/1/$11(00)$/$\frac{6}{6}$}{
\pgfmathsetmacro{\a}{\x *\len};
\pgfmathsetmacro{\b}{\y *\len};
\draw[|-|,thick] (5,\a)-- node[right,color=black](d\x){\z} (5,\b) node[black,left] {\w};
}

\foreach \x/\y/\z/\w in {0/0.16667/$00(11)$/$\frac{1}{6}$,0.16667/0.5/$01$/$\frac{3}{6}$ , 0.5/0.8333/$10$/$\frac{5}{6}$ , 0.8333/1/$11(00)$/$\frac{6}{6}$}{
\pgfmathsetmacro{\a}{\x *\len};
\pgfmathsetmacro{\b}{\y *\len};
\draw[|-|,thick] (3,\a)-- node[right,color=black](d\x){\z} (3,\b) node[black,left] {\w};
}

\foreach \x/\y/\z in {0/0.5/$0$, 0.5/1/$1$}{
\pgfmathsetmacro{\a}{\x *\len};
\pgfmathsetmacro{\b}{\y *\len};
\draw[|-|,thick] (1,\a)-- node[right,color=black](d\x){\z} (1,\b) node[black,left] {\y};
}

\end{tikzpicture}
\caption{Refinement of the output intervals. Round brackets indicate symbols that must follow with probability one.}
\label{fig:ccadm_refine}
\end{figure}
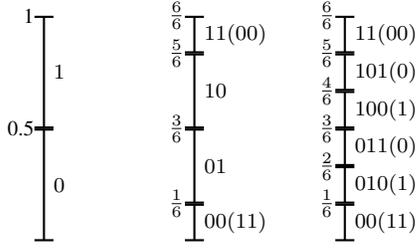
\begin{example}\label{ex:arithmRefine}
Fig. \ref{fig:ccadm_refine} shows a refinement of the output intervals.
Initially there are 2 '0's and 2 '1's in the bag.\
The distribution of the first drawn symbol is
$\pmfd{\bar{A}}{1}(0) =\pmfd{\bar{A}}{1}(1) = \frac{1}{2}$.
When drawing a '0', there are 3 symbols remaining: one '0' and two '1's.
Thus, the probability for a '0' reduces to 1/3 while the probability of '1' is 2/3.
If two '0's were picked, two '1's must follow.
This way we ensure that the encoder output is of the desired type.
Observe that the probabilities of the next symbol conditioned on the previous symbols are unequal in general, i.e, we have
\begin{equation}
P_{\rv{\bar{A}}_2|\rv{\bar{A}}_1}(0|0) \neq P_{\rv{\bar{A}}_2|\rv{\bar{A}}_1}(0|1)
\end{equation} in general.
However, $\pmfn{\bar{A}}{n} = \prod_{i =1}^{n} P_{\rv{\bar{A}}_1|\rv{\bar{A}}^{i-1}}(a_i|a^{i-1})$ is constant on $\nTypePA{n}{\pmf{\bar{A}}}$ as we show in the following proposition.

\end{example}
\begin{proposition}
 After n refinements of the output interval the model used for the refinement step stated above creates equally spaced (equally probable) intervals that are labeled with all sequences in $\nTypePA{n}{\pmf{\bar{A}}}$.
\end{proposition}
\begin{proof}
 All symbols in the bag are chosen at some point.
 Consequently only sequences in $\nTypePA{n}{\pmf{\bar{A}}}$ may appear.
 All possibilities associated with the chosen string are products of fractions $n'_a/n$,
 where $n$ takes on all values from the initial value to $1$ because every symbol is drawn at some point.
 Thus for each string we obtain for its probability an expression that is independent of the realization itself:
 \begin{equation}
 \pmfn{\bar{A}}{n}(a^n) = \frac{n_{a=0}! \dotsb n_{a=k-1}!}{n!} = \frac{1}{\nTypePAcard} \quad \forall a^n \in \nTypePA{n}{\pmf{\bar{A}}}.
 \end{equation}
\end{proof}
Numerical problems for representing the input interval and the output interval occur after a certain number of input bits. For this reason we introduce a \emph{rescaling} each time a new output symbol is known. We explain this next.

\subsection{Scaling input and output intervals}
After we identify a prefix, we are no longer interested in code sequences that do not have that prefix. We scale the input and output interval such that the output interval is [0,1).
Fig.~\ref{fig:ccadm_scaling} illustrates the mapping of intervals  (in$_1$, out$_1$) to (in$_2$, out$_2$).
The refinement for the second symbol works as described in Example \ref{ex:arithmRefine}. If the second input bit is $0$, we know that $10$ must be a prefix of the output. The resulting scaling is shown in Fig.~\ref{fig:ccadm_scaling} as (in$_2$, out$_2$) to (in$_3$, out$_3$).
A more detailed explanation of scaling for arithmetic coding can be found for instance in \cite[Chap.~4]{sayood2006introduction}. We provide an implementation of ccdm online~\cite{website:ccdm}.

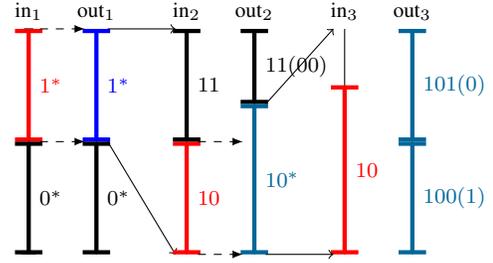
\begin{figure}
\centering
\pgfmathsetmacro{\len}{3.0}

\begin{tikzpicture}[scale={0.3}, font=\footnotesize]
\node [above] at (0,10){in$_1$};
\draw[|-|] (0,0) --(0,10);
\draw (-0.2,10) -- (0.2,10);
\draw[|-|,ultra thick,red](0,5)--(0,10) node[midway,right]{$1^*$};
\draw[|-|,ultra thick](0,0)--(0,5) node[midway,right]{$0^*$};
\node [above] at (3,10){out$_1$};
\draw (3,0) -- (3,10);

\draw[|-|,ultra thick,blue](3,5)--(3,10) node[midway,right]{$1^*$};

\draw[|-|,ultra thick](3,5)--(3,0)node[midway,right]{$0^*$};

\draw[dashed,-latex](0.5,10) --(2.5,10);
\draw[dashed,-latex](0.5,5) --(2.5,5);

\draw[->](3.5,5) -- (6.5,0);
\draw[->](3.5,10) -- (6.5,10);

\node [above] at (7,10){in$_2$};
\draw (7,0) --(7,10);
\draw[|-|,ultra thick,red](7,0)--(7,5) node[midway,right]{$10$};
\draw[|-|,ultra thick](7,5)--(7,10) node[midway,right]{$11$};
\node [above] at (10,10){out$_2$};
\draw (10,0) --(10,10);
\draw[|-|,ultra thick](10,{20/3})--(10,10)node[midway,right]{$11(00)$};
\draw[|-|,ultra thick,blue!60!green](10,{20/3})--(10,0)node[midway,right]{$10^*$};

\draw[dashed,-latex](7.5,5) --(9.5,5);
\draw[dashed,-latex](7.5,0) --(9.5,0);

\node [above] at (14,10){in$_3$};
\draw (14,0) --(14,10);
\draw[|-|,ultra thick,red](14,0)--(14,7.5) node[midway,right]{$10$};
\node [above] at (17,10){out$_3$};
\draw (17,0) --(17,10);
\draw[|-|,ultra thick,blue!60!green](17,5)--(17,10)node[midway,right]{$101(0)$};
\draw[|-|,ultra thick,blue!60!green](17,5)--(17,0)node[midway,right]{$100(1)$};

\draw[->](10.5,0) -- (13.5,0);
\draw[->](10.5,{20/3}) -- (13.5,10);

\end{tikzpicture}
\caption{Scaling of input and output intervals in case the input interval is a subset of an output interval. The latter interval corresponds to $[0,1)$ after scaling. A star indicates that this is just a prefix of the complete word. Round brackets indicate symbols that must follow with probability one.}
\label{fig:ccadm_scaling}
\end{figure}
\section{Conclusion}
 We presented a practical and invertible f2f length distribution matcher that achieves the maximum rate asymptotically in the blocklength.
 In contrast to matchers proposed in the literature \cite{kschischang1993optimal,Ungerboeck2002,bocherer2011matching,amjad2013fixed,Cai2007, bocherer2013arithmetic}  the f2f matcher is robust to synchronization and variable rate problems. Error propagation is limited by the blocklength.
In future work we plan to investigate f2f length codes that perform well in the finite blocklength regime.
 
\section{Acknowledgment}
 We wish to thank Irina Bocharova and Boris Kudryashov for encouraging us to work on the presented approach.

\bibliographystyle{IEEEtran}
\bibliography{IEEEabrv,confs-jrnls,references}

\end{document}